\documentclass{article}
\usepackage{ijcai25}

\usepackage{times}
\usepackage{soul}
\usepackage{url}
\usepackage[hidelinks]{hyperref}
\usepackage[utf8]{inputenc}
\usepackage[small]{caption}
\usepackage{graphicx}
\usepackage{amsmath}
\usepackage{amsthm}
\usepackage{amssymb}
\usepackage{booktabs}
\usepackage{algorithm}
\usepackage{algorithmic}
\usepackage[switch]{lineno}
\usepackage{subcaption}
\usepackage{color}
\usepackage{tabularray}
\usepackage{natbib}

\nolinenumbers

\newtheorem{lemma}{Lemma}
\newtheorem{theorem}{Theorem}

\title{Advancing Differentiable Economics: A Neural Network Framework for Revenue-Maximizing Combinatorial Auction Mechanisms}

\author{
Mai Pham
\and
Vikrant Vaze\And
Peter Chin
\affiliations
Thayer School of Engineering at Dartmouth\\
\emails
mai.p.pham.th@dartmouth.edu,
}

\begin{document}

\maketitle

\begin{abstract}
Differentiable economics, which uses neural networks as function approximators and gradient-based optimization in automated mechanism design (AMD), marked a significant breakthrough with the introduction of RegretNet \citep{regretnet_paper}. It combines the flexibility of deep learning with a regret-based approach to relax incentive compatibility, allowing for approximations of revenue-maximizing auctions. However, applying these techniques to combinatorial auctions (CAs) - where bidders value bundles rather than individual items, capturing item interdependencies - remains a challenge, primarily due to the lack of methodologies that can effectively deal with combinatorial constraints. To tackle this, we propose two architectures: CANet, a fully connected neural network, and CAFormer, a transformer-based model designed to learn optimal randomized mechanisms. Unlike existing methods in traditional AMD, our approach is more scalable and free of assumptions about the structures of allowable bundles or bidder valuations. We demonstrate that our models match current methods in non-combinatorial settings and set new benchmarks for CAs. Specifically, our models consistently outperform benchmark mechanisms derived from heuristic approaches and provide empirical solutions where analytical results are unavailable. This work bridges the gap in applying differentiable economics to combinatorial auctions, offering a scalable and flexible framework for designing revenue-maximizing mechanisms.

\end{abstract}

\section{Introduction}

Mechanism design, a field of economics, focuses on creating incentives and interaction rules among self-interested agents to achieve specific objectives for the group. It differs from traditional game theory by starting with a desired outcome and designing the rules so that agents acting in their self-interest naturally lead to that outcome. A key application of mechanism design is to create auction rules, which play a significant role in economic activities such as the fine art market, advertising on search engines or e-commerce platforms \citep{bajari_2003_ebay,edelman_2007_ads}.

The Vickrey-Clarke-Groves (VCG) mechanism is optimal and strategy-proof when the goal is to maximize total bidder welfare \citep{Vickrey,Clarke,Groves}. However, revenue maximization poses a greater challenge. \cite{Myerson} resolved the revenue-maximizing strategy-proof auction for a single item, later extended by \cite{maskin_1989} to multiple copies of a single item. The general revenue-maximizing auction for multiple items remains unsolved decades later, with only specific two-item cases addressed \citep{armstrong_2000,avery_2000}.

The lack of theoretical progress in revenue-maximizing multi-item auctions inspired the development of Automated Mechanism Design (AMD) \citep{sandholm_2003_amd}. AMD uses computational methods to tailor mechanisms to specific problem instances, with the use of heuristic approaches \citep{Sandholm_2002,Sandholm_2015}, but suffers from the curse of dimensionality in large-scale settings. More recently, leveraging the power of deep learning, differentiable economics uses deep neural networks as function approximators to learn optimal auctions. This approach, introduced with RegretNet \citep{regretnet_paper}, offers a scalable alternative to traditional AMD methods.

Since its introduction, RegretNet has inspired a variety of extensions and improvements \citep{algnet_paper,equivarientnet_paper,regretformer_paper}. The complexity of multi-item auctions is further amplified in combinatorial auctions, where bidders express valuations for bundles rather than individual items, capturing complementarities or substitutabilities. Unlike traditional auctions, the interdependencies between items in combinatorial auctions introduce additional challenges. One such challenge is winner determination, a problem that is well known to be NP-hard, which involves allocating bundles to agents while ensuring that none of the bundles overlap. Although the later part of the RegretNet study discusses the 2-item, 2-bidder auction in a combinatorial setting \citep{regretnet_paper}, the problem is still wide open, primarily due to the lack of methodologies to effectively constrain the allocation space to satisfy combinatorial feasibility.

This paper builds on recent advances in auction theory and machine learning to design combinatorial revenue-optimal auctions. It makes the following contributions:
\begin{enumerate}
    \item It provides an empirical method to find \textbf{near-optimal solutions to revenue-maximizing CA}, addressing problems where analytical solutions are not available. This method is not limited by assumptions about allowable bundle structures or bidder valuations. Similar to the previous work in differentiable economics \citep{regretnet_paper}, the method is approximately strategy-proof.
    \item Unlike previous CA studies, \citep{Sandholm_2015,tacchetti2022combinatorialauction_groves,jamshidi2024diffprivacy}, our models identify \textbf{randomized (lottery) mechanisms}, extending the RegretNet family to CA problems. It is well known that randomization can increase revenue in CAs. Our empirical experiments show substantially higher performance than deterministic mechanisms.
    \item It introduces two architectures: \textbf{CANet} and \textbf{CAFormer}, extending prior designs from non-combinatorial auctions to a CA setting. While CANet is sensitive to the order and structure of the input, CAFormer is permutation-equivariant, offering better generalization in scenarios with varying input configurations. In addition, it implements several techniques to improve stability in training and tackle the vanishing gradient problem.
    \item It provides an example of successful use of \textbf{gradient based methods} for \textbf{constrained optimization problems}, which are known to be challenging. In our particular case, the constraints are related to combinatorial feasibility. This approach can be extended to other differentiable combinatorial optimization problems.
\end{enumerate}

Section \ref{sec:CA_design} introduces the problem statement. Section \ref{sec:method} presents the proposed method. Section \ref{sec:results} provides the results and Section \ref{sec:discussion} concludes with a discussion and the next steps.

\section{Combinatorial Auction Design}\label{sec:CA_design}
\subsection{Optimal CA as a linear program}
\paragraph{General auction setting} We study a setting with one seller, a set $N$ of $n$ bidders, and a set $M$ of $m$ items for sale. In an auction, bidders submit bids for the bundles of items given the auction rule. An auction rule is a pair: an allocation rule $g(.)$ and a payment rule $p(.)$, where $g_i$ is the bundle of items received by the bidder $i$ and $p_i$ is the payment by the bidder $i$.

\paragraph{Non-additive valuations} We define $K$ as the set (indexed by $S$) of bundles of size $k = 2^{m} -1$, which comprises \textit{all possible non-empty combinations} of items for each bidder. The \textit{valuation profile} $v = (v_1, v_2, \ldots, v_n)$ is private information and is unknown to the auctioneer. Let $V_i$ denote the set of all possible valuation functions for the bidder $i$. In a CA, the\textit{ valuation function} for the bidder $i$ is given by the vector $v_i = (v_{i1}, v_{i2}, \ldots, v_{ik})$, where each element specifies the value that bidder $i$ attaches to a certain bundle $S \in K$ of items from $M$. Each valuation function $v_i$ is sampled from a continuous probability density $F_i$, and $F_i$ is positive on all $V_i$. The auction is \textit{symmetric} if $F_i = F \: \forall i \in N$, and \textit{asymmetric} if $F_i \neq F_j$ for $i \neq j$.

\paragraph{Quasi-linear utilities} We make the standard assumption that each bidder $i$ has a \textit{quasi-linear utility function}, that is, a bidder $i$ with valuation $v_i$ for the bid profile $b$ receives utility $u_i(v_i; b) = v_i(g_i(b)) - p_i(b)$, where $v_i(g)$ is bidder $i$'s valuation for allocation $g$, $p_i$ is the payment by the bidder $i$.

\paragraph{Incentive Compatibility and Regret} 
A mechanism is a \textit{dominant-strategy incentive compatible (DSIC)} mechanism if each bidder's utility is maximized when reporting true valuation, regardless of what other bidders report.
\begin{equation}
    \label{dsic}
    u_i (v_i; (v_i, b_{-i})) \ge u_i (v_i; (b_i, b_{-i})) \: \forall b_i \ge 0,\: \forall i \in N
\end{equation}

To quantify the violation of DSIC, \textit{expected ex-post regret} for each bidder $i$ is defined by the following equation.
\begin{equation}
    \text{rgt}_i = \mathbf{E} \left[ \max_{v'} u_i (v_i; (v', v_{-i})) - u_i (v_i; (v_i, v_{-i}))) \right]
\end{equation}

\paragraph{Individual Rationality (IR)} A mechanism is \textit{ex-post individual rational} if each bidder receives a nonnegative utility when participating truthfully. 
\begin{equation}
    u_i (v_i; (v_i, b_{-i})) \ge 0 \: \forall b_{-i}, \: \forall i \in N
    \label{ir}
\end{equation}

\paragraph{Combinatorial Feasibility}

Let $z_{iS}$ specify if a bundle \( S \) is allocated to bidder \( i \) (\( z_{iS} = 1 \)) or not (\( z_{iS} = 0 \)). The standard CA formulation requires the following constraints to hold:
\begin{equation}
    \sum_{i\in N} \sum_{S \in K: j\in S} z_{iS} \leq 1 \quad \forall j \in M
    \label{item_sum}
\end{equation}
\begin{equation}
    \sum_{S \in K} z_{iS} \leq 1 \quad \forall i \in N
    \label{agent_sum}
\end{equation}
\begin{equation}
    z_{iS} \in \{0, 1\} \quad \forall i \in N, S \in K
    \label{binary}
\end{equation}

Constraint \eqref{item_sum} ensures that each item \( j \) is allocated to at most one bidder, implicitly guaranteeing that no two bidders can receive overlapping bundles. 
Constraint \eqref{agent_sum} ensures that each bidder \( i \) is assigned at most one bundle.

To make the solution space differentiable, we relax the binary constraint \eqref{binary} by the following theorem:

\begin{theorem}[Birkhoff, 1916]
    Every doubly stochastic matrix is a convex combination of permutation matrices.  
    \label{theoBirkhoff}
\end{theorem}

This ensures that a randomized allocation represented by a doubly-stochastic matrix can be decomposed into a lottery over deterministic one-to-one assignments. Note that since each bundle consists of a unique subset of items, constraint \eqref{item_sum} implies that each bundle can be allocated at most once. So relaxing the binary constraint \eqref{binary}
\begin{equation}
    z_{iS} \in [0, 1] \quad \forall i \in N, S \in K
    \label{real}
\end{equation}
enables the representation of lottery decomposable randomized allocations.

We refer to the inequalities \eqref{item_sum}, \eqref{agent_sum}, and \eqref{real} as \textit{combinatorial constraints}. Any allocation satisfying these constraints is considered \textit{combinatorially feasible}.

\paragraph{Revenue Maximization}
The primary objective of the auctioneer in a CA is to maximize revenue while respecting the constraints of the mechanism. For a given valuation profile \( v \), the revenue for a given payment rule \( p(v) \) is defined as:
\[
\text{rev} = \sum_{i \in N} p_i(v).
\]
The revenue maximization problem in CAs can be formulated as a linear program over the allocation and payment rules. The linear program is expressed as:
\[
\max_{\textbf{g}, \textbf{p}} \quad \sum_{i \in N} p_i(v)
\]
subject to \eqref{dsic}, \eqref{ir}, \eqref{item_sum}, \eqref{agent_sum}, \eqref{real} and \( p_i(v) \ge 0, \: \forall i \in N \). For brevity, in the sections below, we use the same notations $u, \text{rev}, \text{rgt}$ for empirical quantities in the learning problem.

\subsection{Auction Design as a learning problem} \label{regretnet_intro}
RegretNet \citep{regretnet_paper} provides a deep learning framework for optimal auction design, addressing multi-bidder, multi-item scenarios where analytical solutions are unknown. It encodes auction allocation and payment rules within neural network weights.

The architecture comprises two networks: the allocation network \( \mathbf{A}^{\text{net}} \) and the payment network \( \mathbf{P}^{\text{net}} \). Both process a bid vector \( \mathbf{b}^{nm} \) through fully-connected layers. The Allocation Network maps bids to allocation probabilities \( \mathbf{A}^{\text{net}}(\mathbf{b}^{nm}) = \mathbf{Z}^{nm} \), where \( z_{ij} = \frac{e^{\tilde{z}_{ij}}}{\sum_{i=1}^{n+1} e^{\tilde{z}_{ij}}} \) and allows an item to remain unallocated by introducing a dummy participant. The payment network maps bids to payment allocations \( \mathbf{P}^{\text{net}}(\mathbf{b}^{nm}) = \tilde{\mathbf{p}}^n \), where \( \tilde{p}_i \) is scaled using sigmoid, and \( p_i = \tilde{p}_i \sum_{j=1}^m z_{ij}b_{ij} \), adhering to IR constraints \eqref{ir}.

RegretNet minimizes empirical negative revenue plus the regret penalty over profiles \( V \):
\begin{equation}
    \min_{\mathbf{w, \mathbf{\lambda}}, \rho} \left( - \sum_{i \in N} \mathbb{E}[P_i(v; \mathbf{w})] + \lambda_i \text{rgt}_i + \frac{\rho}{2} \text{rgt}_i^2 \right)
    \label{loss_regretnet}
\end{equation}
where \(\text{regret } \text{rgt}_i(v_i', v; \mathbf{w}) = \max_{v_i'}(u(v_i, (v_i', v_{-i}); \mathbf{w}) - u(v_i, v; \mathbf{w}))\). Regret is iteratively minimized using augmented Lagrangian and gradient descent:
\begin{equation}
\mathcal{L}_{\text{inner}}(v_i') = -u(v_i, (v_i', v_{-i}); \mathbf{w})
\label{inner_loss}
\end{equation}
and the Lagrange multipliers \(\lambda_i\) and \(\rho\) are updated as \(\lambda_i \leftarrow \lambda_i + \rho \tilde{R}_i\), \(\rho \leftarrow \rho + \rho \Delta\).

The results of RegretNet are further extended to address combinatorial auctions in a two-item, two-bidder setting by adding the constraint:

\[
\forall i, \quad z_{i, \{1\}} + z_{i, \{2\}} \leq 1 - \sum_{i=1}^{n} z_{i, \{1,2\}}
\]

To accommodate this constraint, two softmax layers are used; one outputs a set of bidder scores, and the other a set of item scores, denoted by $\bar{s}_{i,S}$ and $\bar{s}_{i,S}^{(j)}$. The set of item scores is then used to compute a normalized set of scores for all $i$ and $S$, given by
\[
\bar{\tilde{s}}_{iS} = 
\begin{cases} 
\bar{s}_{iS}^{(i)} & S = \{1\}, \{2\} \\ 
\min \{\bar{s}_{iS}^{(k)}\} & S = \{1, 2\}
\end{cases}
\]
and the final allocation \(z_{i,S}\) is determined by the minimum of the bidder score and normalized item score.
\[
\quad z_{iS} = \min \left( \bar{s}_{iS}, \bar{\tilde{s}}_{iS} \right)
\]

\begin{figure*}[ht]
    \centering
    \begin{subfigure}{\textwidth}
        \centering
        \includegraphics[width=\textwidth]{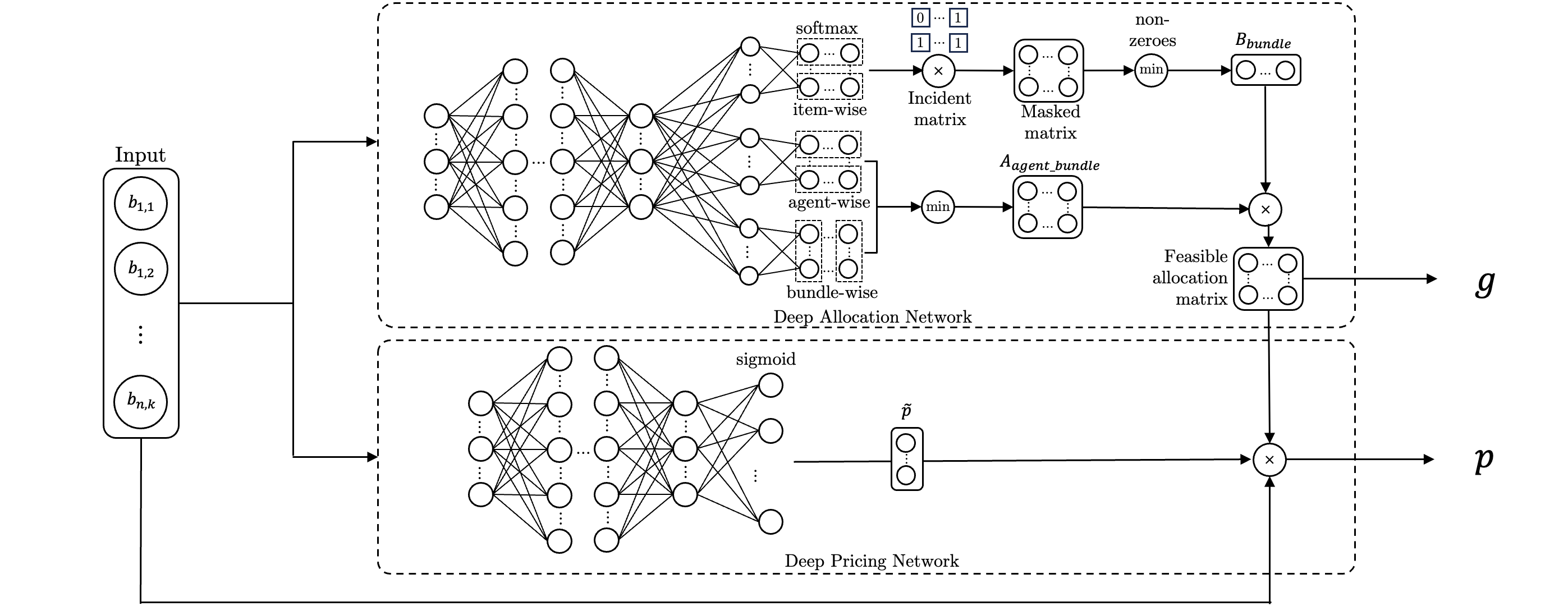}
        \caption{The CANet architecture for $n$ bidders and $m$ items takes $n \times k$ bid vectors as input. It consists of two fully connected network modules: the allocation network, which outputs a feasible allocation matrix $g$, and the pricing network, which generates a pricing vector $\tilde{p}$.}
        \label{fig:CANet_architecture}
    \end{subfigure}
    \hfill
    \begin{subfigure}{\textwidth}
        \centering
        \includegraphics[width=\textwidth]{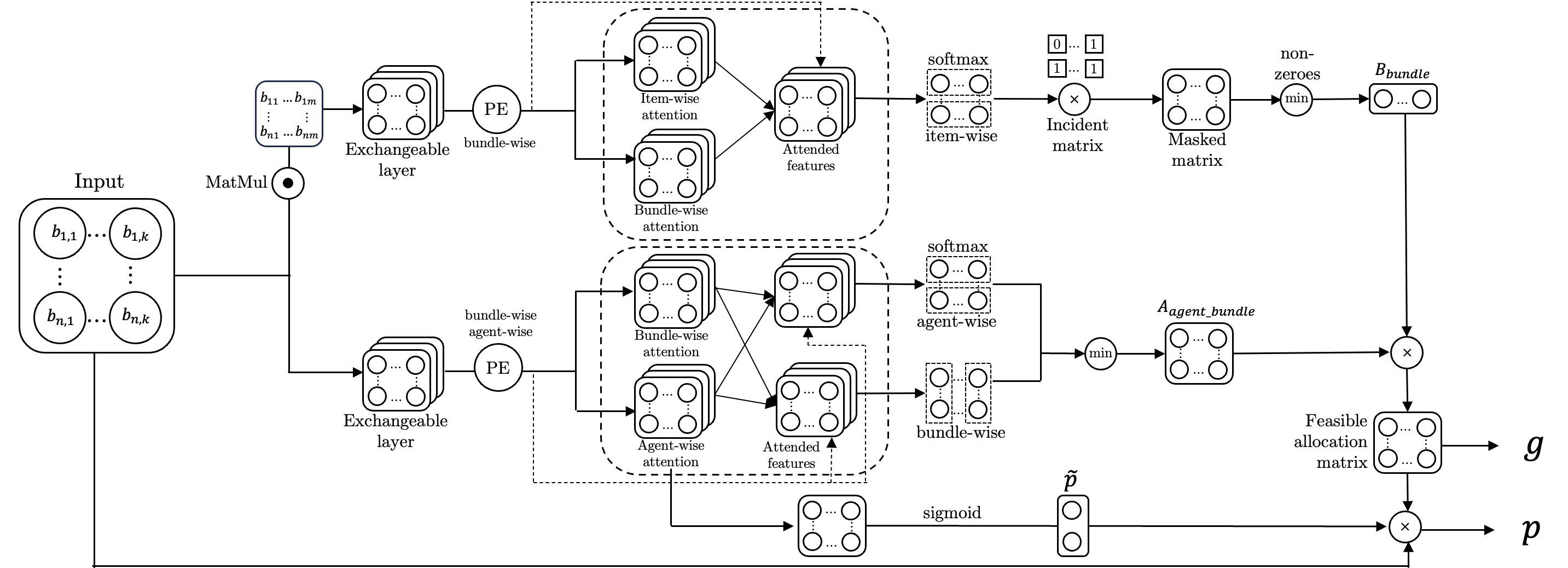}
        \caption{The CAFormer architecture processes $n \times k$ bid matrices using attention mechanism to encode permutation-equivariant item-bundle and agent-bundle representations. It outputs a feasible allocation matrix $g$, with the pricing vector $\tilde{p}$ derived from the agent embedding.}
        \label{fig:CAFormer_architecture}
    \end{subfigure}
    \caption{CANet and CAFormer Architectures.}
    \label{fig:comparison_architectures}
\end{figure*}

\subsection{Related Work}

The seminal works of \cite{Vickrey} introduced auction mechanisms for single items, which were later extended to multi-item settings by developing the Vickrey-Clarke-Groves (VCG) mechanism \citep{Vickrey,Clarke,Groves}. Although VCG achieves efficiency in dominant strategies, it often fails to maximize revenue. \cite{Myerson} laid the theoretical foundation for revenue-maximizing auctions.

The study of combinatorial auctions originated from the need to allocate resources among bidders with complex preferences over bundles of items. \cite{cramton2005survey} provide a detailed review of combinatorial auctions. Optimal combinatorial auctions are NP-hard because of the exponential number of possible bundles and the Incentive Compatibility (IC) and IR constraints. Heuristic and exact methods have been studied to solve the winner determination problem as a combinatorial optimization problem \citep{Sandholm_2002,Wu_2016_clique,Ray_2021}. In the early work, both the winner determination and pricing problems were tackled using integer programming and approximation algorithms. \cite{laviswamy} developed polynomial-time algorithms for approximately optimizing VCG-based mechanisms in certain structured settings. Another direction, extending the results of \cite{roberts1979}, is known as Affine Maximizer Auctions (AMAs). These are variations of the VCG mechanism that adjust the allocation process by assigning positive weights to each bidder's welfare and incorporating boosts for different allocations, potentially leading to higher revenue than the VCG mechanism \citep{Sandholm_2015,Curry_2023_AMA}.

Recent advances in differentiable economics have extended its applications to various domains, reflecting the growing synergy between machine learning and economic theory. These works leverage neural architectures and differentiable approaches to address classical and emerging problems in auction design and beyond. Several papers have extended the ideas introduced in RegretNet. \cite{algnet_paper} model optimal auction as a two-player game and proposes a loss function with fewer hyperparameters. \cite{equivarientnet_paper} and \cite{regretformer_paper} improve the results using different architectures. Others have applied this framework to auctions with fairness or budget constraints \citep{Kuo2020ProportionNet,feng2018budget}, or adapted similar methods to tackle broader mechanism design challenges \citep{Golowich_2018,ravindranath2023twosided}. Several papers leverage machine learning-based approaches in combinatorial auctions. \cite{gianluca_elicite_2018} used a machine learning-based elicitation algorithm to identify which values to query from the bidders. \cite{tacchetti2022combinatorialauction_groves} propose a network architecture to learn Groves payment rules in combinatorial auctions with certain bidding languages. \cite{jamshidi2024diffprivacy} present a machine learning-powered combinatorial auction based on the principles of differential privacy. \cite{brero2021iterativecombinatorial} present a machine learning-powered iterative combinatorial auction. \cite{ravindranath2024RLsequential} use deep Reinforcement Learning for sequential combinatorial auctions.




\section{Proposed Method} \label{sec:method}
After formulating CAs as a learning problem, we will propose models that facilitate end-to-end learning through gradient-based optimization techniques.
\subsection{Constraint Enforcement}
\label{subsec:constraint_enforcement}
Recall from Section \ref{sec:CA_design} that the allocation must satisfy constraints \eqref{item_sum}, \eqref{agent_sum} and \eqref{real}. We decompose the allocation matrix \( \mathbf{Z} \in \mathbb{R}^{n \times k} \) into the product of two matrices:
\[
\mathbf{Z} = \mathbf{B}^{\text{bundle}} \cdot \mathbf{A}^{\text{agent-bundle}},
\]
where 
$\mathbf{A}^{\text{agent-bundle}}$ and $\mathbf{B}^{\text{bundle}}$ are defined below.   This approach is motivated by a two-step allocation process: first, the allocation of items to the bundles, satisfying the item-wise constraint, and then the allocation of the bundles to the bidders, satisfying the other constraints. Each entry in the final allocation matrix represents the probability of allocating a bundle to a bidder, which is calculated by the product of the probability that the bundle is allocated $B^{\text{bundle}}_{S}$ and the probability that the bidder receives the bundle, given that the bundle is allocated $A^{\text{agent-bundle}}_{iS}$.

The item-to-bundle allocation matrix \( \mathbf{B}^{\text{bundle}} \) is derived from the output of the Allocation Network, normalized to represent item-wise probabilities. To construct \( \mathbf{B}^{\text{bundle}} \), we define the incident matrix \( \mathbf{I} \in \mathbb{R}^{m \times k} \), where:
\[
I_{iS} = 
\begin{cases}
1 & \text{if item } i \text{ is included in bundle } S, \\
0 & \text{otherwise}.
\end{cases}
\]

The initial bundle allocation matrix \( \mathbf{B} \in \mathbb{R}^{m \times k} \) is adjusted by mapping bundles to items with a softmax function across $m$ items:
\[
\mathbf{B}^{\text{adjusted}} = \text{softmax}_M (\mathbf{B} \cdot \mathbf{I}).
\]

To ensure valid item-to-bundle allocations, non-positive values in \( \mathbf{B}^{\text{adjusted}} \) are replaced with a large constant \( M \):
\[
\mathbf{B}^{\text{masked}} = 
\begin{cases}
\mathbf{B}^{\text{adjusted}}, & \text{if } \mathbf{B}^{\text{adjusted}} > 0, \\
M, & \text{otherwise}.
\end{cases}
\]

For each item \( i \), we compute the minimum value across all bundles and reshape \( \mathbf{B}_{\text{min}} \) to include the agent dimension:
\[
\mathbf{B}^{\text{bundle}} = \text{Unsqueeze} (\min_{S} \mathbf{B}^{\text{masked}}),
\]

The bundle-to-agent allocation matrix \( \mathbf{A}_{\text{agent-bundle}} \) is computed from the allocation network outputs $\mathbf{A}$ and $\mathbf{A'}$ using the softmax function along the agent and the bundle dimension.
\[
\mathbf{A}^{\text{agent}} = \text{softmax}_N (\mathbf{A}), \quad \mathbf{A}^{\text{bundle}} = \text{softmax}_K(\mathbf{A'})
\]
then taking element-wise minimum of the two matrices
\[
\mathbf{A}^{\text{agent-bundle}} = \min(\mathbf{A}^{\text{agent}}, \mathbf{A}^{\text{bundle}})
\]
The final allocation matrix \( \mathbf{Z} \) is given by:
\[
\mathbf{Z} = \mathbf{B}^{\text{bundle}} \cdot \mathbf{A}^{\text{agent-bundle}},
\]
where each entry \( z_{iS} \) represents the probability of allocating bundle \( S \) to agent \( i \).

\begin{lemma}
    Matrix \( \mathbf{Z} \) is combinatorially feasible.
\end{lemma}

\begin{proof}

The matrix \( \mathbf{A}^{\text{agent-bundle}} \) represents the bundle-to-agent allocation. Each entry \( A^{\text{agent-bundle}}_{iS} \) is computed as:
\[
A^{\text{agent-bundle}}_{iS} = \min 
\left(\dfrac{e^{a_{iS}/\theta}}{\sum_{i'} e^{a_{i'S}/\theta}},  
\dfrac{e^{a'_{jS}/\theta}}{\sum_{S'} e^{a'_{iS'}/\theta}}
\right).
\footnote{In all of our implementations, the softmax function incorporates a \textit{temperature} parameter \( \theta \in \mathbb{R}^+ \), inspired by the Boltzmann distribution in statistical mechanics, which controls the concentration of probability mass around the highest logit.}
\]
This ensures that
\begin{equation}
    \sum_{i} A^{\text{agent-bundle}}_{iS} \le 1, \forall S \in K
    \label{agent_sum_ex}
\end{equation}
\begin{equation}
    \sum_{S} A^{\text{agent-bundle}}_{iS} \le 1, \forall i \in N
    \label{bundle_sum_ex}
\end{equation}
From the definition of \( z_{iS} \), the allocation of any bundle \( S \) is limited by its least-available item:
\[
\sum_{i} \sum_{S \ni j} z_{iS}
\leq \sum_{i} \sum_{S \ni j} \min_{j \in S} \{ B^{\text{adjusted}}_{jS}  \} \cdot A^{\text{agent-bundle}}_{iS}.
\]
Since \( B^{\text{adjusted}} \) is normalized with softmax along item-wise dimension to ensure that no item \( j \) is allocated more than once:
\[
\sum_{S \ni j} \min_{j \in S} \{ B^{\text{adjusted}}_{jS} \} \leq 
\sum_{S\in K} \min_{j \in S} \{ B^{\text{adjusted}}_{jS} \} \leq 1, \quad \forall j \in M.
\]
Thus,
\[
\sum_{i} \sum_{S \ni j} z_{iS} \leq \sum_{i} \sum_{S \ni j} A^{\text{agent-bundle}}_{iS}.
\]
For a fixed item \( j \), we can write:
\[
\sum_{i} \sum_{S \ni j} A^{\text{agent-bundle}}_{iS} \leq \sum_{S \ni j} \sum_{i} A^{\text{agent-bundle}}_{iS}.
\]
From \eqref{agent_sum_ex}, we have:
\[
\sum_{i} \sum_{S \ni j} A^{\text{agent-bundle}}_{iS} \leq \sum_{S \ni j} 1.
\]
From \eqref{bundle_sum_ex}, the number of bundles \( S \) containing \( j \) is at most the total number of bundles \( k \), constraint \eqref{item_sum} is satisfied:
\[
\sum_{i} \sum_{S \ni j} A^{\text{agent-bundle}}_{iS} \leq 1.
\]
Thus, \( \mathbf{Z} \) is combinatorially feasible.
\end{proof}

\subsection{CANet Architecture}
Motivated by RegretNet \citep{regretnet_paper}, we present CANet, a deep neural network architecture designed for combinatorial auction mechanisms. The CANet architecture, visualized in Figure \ref{fig:CANet_architecture}, comprises two modules: the Deep Allocation Network and the Deep Pricing Network. The input to CANet, denoted by \( \mathbf{b} \), is a vector \( \mathbf{b} \) \( \in \mathbb{R}^{n \cdot k} \), where \( b_{iS} \sim F_i\) is the bid submitted by agent \( i \) for bundle \( S \).

The Deep Allocation Network computes a feasible allocation matrix \( \mathbf{Z} \in \mathbb{R}^{n \times k} \) ensuring that both individual item constraints and global allocation constraints are satisfied. The input \( \mathbf{b} \) is passed through fully connected layers with non-linear activations. The network produces three output vectors $\mathbf{A} \in \mathbb{R}^{n \times k}$ and $\mathbf{A'} \in \mathbb{R}^{n \times k}$ representing an unnormalized agent-bundle allocation; and \( \mathbf{B} \in \mathbb{R}^{m \times k} \) representing an unnormalized bundle-item allocation. These outputs are combined following the process described in section \ref{subsec:constraint_enforcement} to form a feasible allocation matrix $\mathbf{Z}$.

The Deep Pricing Network computes the pricing vector \( \tilde{\mathbf{p}} \in \mathbb{R}^n \), which represents a discount applied to the bids. The input \( \mathbf{b} \) is processed through fully connected layers with sigmoid transformation, to ensure $\tilde{p}_i \in [0, 1]$. The final price is vector \( \mathbf{p} \) with \( p_i = \tilde{p}_i \sum_{S=1}^k z_{iS}b_{iS} \).

\subsection{CAFormer Architecture}

To enhance expressivity and ensure permutation-equivariance and context-awareness, we introduce CAFormer. This architecture combines exchangeable layers (proposed in \cite{deepsets} and used in EquivariantNet \cite{equivarientnet_paper}) with stacked attention layers \citep{attention_is_all_you_need}, as seen in RegretFormer \citep{regretformer_paper}. The architecture is illustrated in Figure \ref{fig:CAFormer_architecture}. The attention mechanism is permutation-invariant, while exchangeable layers ensure permutation-equivariance—properties particularly desirable in symmetric auctions. For asymmetric settings or preserving input order across the bundle dimension, we apply agent-wise and bundle-wise positional encoding as presented in \citep{attention_is_all_you_need} after the exchangeable layers.

The input to CAFormer is a two-dimensional matrix \( \mathbf{b} \in \mathbb{R}^{n \times k} \). Unlike CANet that can output an item-bundle allocation at the last layer, CAFormer must be designed properly to maintain insensitivity to the size of the problem.  We achieve it by transforming the bid matrix into a representation of item-bundle allocation $\mathbf{b}' \in \mathbb{R}^{m \times k} $ through matrix multiplication with the individual item bid matrix $\mathbf{b}^{\text{item}} \in \mathbb{R}^{n \times m}$.
\[
\mathbf{b}' = \text{MatMul}(\mathbf{b}^\top, \mathbf{b}^{\text{item}})^\top,
\]
where \(^\top \) denotes the transpose operation.
Both $\mathbf{b}$ and $\mathbf{b}'$ are then expanded in a feature dimension that contain global information of bids with exchangeable layers, a key building block in architectures like EquivariantNet \citep{equivarientnet_paper} and RegretFormer \citep{regretformer_paper} which ensures permutation-equivariance. This layer processes a tensor \( \mathbf{X} \in \mathbb{R}^{n \times  k \times 1 }\), into an output tensor \( \mathbf{X}_\text{ex} \in \mathbb{R}^{n \times  k \times d} \). The exchangeable layer computes the element of the output channels by aggregating input tensor \( B \) over elements, rows, columns, and globally, weighted by learnable parameters, and applies a non-linear activation \( \sigma \). Details can be found in Deep Sets \citep{deepsets}.
\[
\mathbf{b}_\text{ex} = \text{Exchangeable}(\mathbf{b}), \quad \mathbf{b}_\text{ex}' = \text{Exchangeable}(\mathbf{b}')
\]
After each exchangeable layer, we apply sequential attention-based blocks to the tensors \( \mathbf{b}_{\text{ex}} \) and \( \mathbf{b}_\text{ex}' \). Each block consists of multi-head self-attention layers with residual connections. 

For each layer, the input tensor is reshaped to enable item-wise, bundle-wise, or agent-wise self-attention. Item-wise attention operates on reshaped input \( \mathbf{b}'_{\text{item}} \in \mathbb{R}^{m \times d} \), bundle-wise attention operates on \( \mathbf{b}_{\text{bundle}} \text{ and } \mathbf{b}'_{\text{bundle}} \in \mathbb{R}^{k \times d} \), and agent-wise attention operates on \( \mathbf{b}_{\text{agent}} \in \mathbb{R}^{n \times d} \), where \( d \) is the feature dimensionality of the input. The outputs are concatenated, forming \( \mathbf{b}_{\text{concat}}' \in \mathbb{R}^{m \times k \times d_{\text{concat}}'} \) and \( \mathbf{b}_{\text{concat}} \in \mathbb{R}^{n \times k \times d_{\text{concat}}} \), where \( d_{\text{concat}} \text{ and } d_{\text{concat}}' \) are the combined attended features. After the attention-blocks, we use fully-connected layers \( f_{\text{fc}}^{\text{item}} \), \(f_{\text{fc}}^{\text{bundle}}\), and\( f_{\text{fc}}^{\text{agent}} \) to map the features back to the original dimensionality \( 1 \):
\[
\mathbf{B} = f_{\text{fc}}^{\text{item}}
(\mathbf{b}_{\text{concat}}'),
\]
\[
\mathbf{A} = f_{\text{fc}}^{\text{agent}}(\mathbf{b}_{\text{concat}}), \quad
\mathbf{A'} = f_{\text{fc}}^{\text{bundle}}(\mathbf{b}_{\text{concat}})
\]
Again, following the procedure outlined in Section \ref{subsec:constraint_enforcement}, we derive a feasible allocation matrix \( \mathbf{Z} \). The pricing vector \( \tilde{\mathbf{p}} \in \mathbb{R}^n \) is derived by applying a fully connected layer with sigmoid activation to the agent-wise attention output \( \mathbf{b}_{\text{agent}} \), reducing its dimensionality to 1. The final prices are computed using the same method as in CANet.

\subsection{Training Procedure}

Similarly to RegretNet, our training process also includes an inner utility maximization \eqref{inner_loss} and an outer revenue maximization and regret minimization \eqref{loss_regretnet}. For brevity, we refer the reader to the original paper by \cite{regretnet_paper} for details of the adversarial training. 
Previous works \citep{regretnet_paper,equivarientnet_paper,regretformer_paper} relies heavily on gradient-based Lagrangian optimization (see Section \ref{regretnet_intro}). This method is known to be sensitive to the choice of hyper-parameters. An example is in the outer loss \eqref{loss_regretnet}, when \(\lambda_i\) becomes excessively large, the model disproportionately focuses on minimizing regret, and gradients for regret become negligible due to a scaling effect or poor numerical conditioning. To ensure a stable tradeoff, we constrain the loss weights such that their sum is 1. This simple balancing mechanism does not directly tackle the overemphasis on regret, but it ensures regret gradient does not shrink due to the excessive weight $\lambda_i$. Additionally, we apply a logarithmic transformation to the revenue term to avoid unbounded growth. This is in the spirit of multi-task learning, which seeks to optimize both objectives in the present of trade-offs. Our outer loss function is defined as:

\[
\mathcal{L}_{\text{outer}} = - w_{rev} \log(1 +\text{rev}) + w_{rgt}\text{rgt},
\]
where \(w_{rev}\) and \(w_{rgt}\) are task weights that represent the emphasis on revenue and regret, respectively. 


Following \cite{regretformer_paper}, we also set a regret budget, allowing controlled violations of DSIC by gradually decreasing the budget during training. The update for \(w_{rgt}\) is explicitly computed as:
\begin{equation}
   g_t = \log(\text{rgt})  - \log(\bar{\text{rgt}})-  \log(1 + \alpha\text{rev})
   \label{eq:rgt_update}
\end{equation}
where \(\bar{\text{rgt}}\) is a target regret value. The parameter $\alpha$ is optional and adjusts the ratio between the target regret and revenue.

To ensure proper balance between revenue and regret, we normalize the weights with tanh function scaled with paramter $\rho$.
Our $w_{rev}$ and $w_{rgt}$ are maintained in range $[0, 1]$. Using exponential moving averages for the first and second moments of \(g_t\), we update \(w_{rgt}\) with a bias-corrected Adam-style rule with learning rate $\gamma$. We explicitly present the update rule for the weights in Algorithm \ref{alg:update_rule}.

\begin{algorithm}[tb]
    \caption{Adam Update Rule for $w_{rev}$ and $w_{rgt}$}
    \label{alg:update_rule}
    \begin{algorithmic}[1]
        \STATE Initialize $m_w \gets 0$, $v_w \gets 0$, $\beta_1$, $\beta_2$, $\epsilon$, $\gamma_w, \rho$
        
        \STATE Compute gradient $g_t$ for $w_{rgt}$ with equation (\ref{eq:rgt_update})

        \STATE Update first and second moments
        \STATE \quad $m_w \gets \beta_1 m_w + (1 - \beta_1) g_t$
        \STATE \quad $v_w \gets \beta_2 v_w + (1 - \beta_2) g_t^2$
        
        \STATE Bias correction
        \STATE \quad $\hat{m}_w \gets \dfrac{m_w }{1 - \beta_1^{\epsilon + 1}}; \quad \hat{v}_w \gets \dfrac{v_w}{ (1 - \beta_2^{\epsilon + 1}}$

        \STATE Update $w_{rgt}$
        \STATE \quad $w_{rgt} \gets w_{rgt} + \gamma_w \dfrac{\hat{m}_w }{ (\sqrt{\hat{v}_w} + \epsilon)}$

        \STATE Normalize weights
        \STATE \quad $w_{rgt} = \max \left(\frac{e^{w_{rgt}/\rho} - e^{-w_{rgt}/\rho}}{e^{w_{rgt}/\rho} + e^{-w_{rev}/\rho}}, 0\right)$
        \STATE \quad $w_{rev} \gets 1 - w_{rgt}$

        \STATE return $w_{rev}, w_{rgt}$

    \end{algorithmic}
\end{algorithm}

\section{Experimental Results} \label{sec:results}

\begin{table*}
\centering
\resizebox{\textwidth}{!}{%
\begin{tblr}{
  width = \textwidth,
  colspec = {Q[110]Q[67]Q[67]Q[67]Q[67]Q[67]Q[67]Q[67]Q[67]Q[67]Q[67]Q[77]Q[67]},
  row{1} = {c},
  row{2} = {c},
  cell{1}{2} = {c=2}{0.134\textwidth},
  cell{1}{4} = {c=2}{0.134\textwidth},
  cell{1}{6} = {c=2}{0.134\textwidth},
  cell{1}{8} = {c=2}{0.134\textwidth},
  cell{1}{10} = {c=2}{0.134\textwidth},
  cell{1}{12} = {c=2}{0.144\textwidth},
  hline{1,3,12} = {-}{},
}
          & 2x2 (B) &       & 2x3 (B) &       & 2x5 (B) &       & 2x2 (C) &       & 2x3 (C) &       & 2x5 (C) &       \\
          & Rev     & Rgt   & Rev     & Rgt   & Rev     & Rgt   & Rev     & Rgt   & Rev     & Rgt   & Rev     & Rgt   \\
VCG       & 2.405   & 0     & 3.537   & 0     & 5.838   & 0     & 2.847   & 0     & 4.239   & 0     & 6.987   & 0     \\
AMA       & 2.760   & 0     & -       & -     & -       & -     & 4.240   & 0     & -       & -     & -       & 0     \\
VVCA      & 2.770   & 0     & -       & -     & -       & -     & 4.240   & 0     & -       & -     & -       & 0     \\
BLAMA     & 2.630   & 0     & 4.125   & 0     & 7.280   & 0     & 4.080   & 0     & 4.812   & 0     & 8.164   & 0     \\
ABAMA     & 2.630   & 0     & 4.166   & 0     & 7.145   & 0     & 4.010   & 0     & 5.916   & 0     & 11.140  & 0     \\
BBBVVCA   & 2.620   & 0     & 4.105   & 0     & 7.156   & 0     & 4.010   & 0     & 5.898   & 0     & 11.135  & 0     \\
RegretNet & 2.871   & 0.001 & -       & -     & -       & -     & 4.270   & 0.001 & -       & -     & -       & -     \\
CANet     & 2.904   & 0.001 & 4.369   & 0.002 & 7.304   & 0.007 & 4.285   & 0.001 & 6.556   & 0.001 & 11.393  & 0.002 \\
CAFormer  & \textbf{2.919}   & 0.001 & \textbf{4.388}   & 0.002 & \textbf{7.318}   & 0.003 & \textbf{4.403}   & 0.002 & \textbf{6.693}   & 0.001 & \textbf{11.534}  & 0.004 
\end{tblr}
}
\caption{Combinatorial results: Grid search methods (AMA and VVCA) become computationally infeasible in higher dimensions. RegretNet results are limited to \( 2 \times 2 \) cases. CANet and CAFormer beat the benchmarks in all settings. Both CAFormer and CANet achieve negligible regret. But CAFormer consistently outperforms CANet in revenue.}
\label{combinatorial_results}
\end{table*}

\paragraph{Set up} 
Our framework is implemented in PyTorch and trained for 50,000 or 100,000 outer optimization iterations, depending on the problem size. 
All networks used Glorot uniform initialization \citep{glorot_xavier} and tanh activations. 
We sampled 640,000 valuation profiles offline for training and 10,000 profiles online for testing. Each outer optimization update included 50 inner steps during training and 1,000 during validation. Typical hyperparameters include: tanh scaler $\rho = 2$, softmax temperature $\theta \in \{10, 15, 25\}$, network learning rates $\{0.0004, 0.0007\}$, regret learning rate $\gamma \in \{0.01, 0.02\}$ and the target regret-revenue adjustment factor $\alpha \in \{0.5, 1\}$. CANet’s allocation and pricing networks use $\{3, 6\}$ layers with 100 hidden nodes each, while CAFormer has a single attention layer with 2 heads and 64 hidden features. We initialized $w_{\text{rgt}} = 1$, exponentially annealling regret targets from $\bar{\text{rgt}}_{\text{start}} = 0.05$ to $\bar{\text{rgt}}_{\text{end}} \in \{0.0008, 0.001, 0.002, 0.003\}$ in $2/3$ of the training iterationw.  Experiments are conducted on three scales: 2 bidders, 2 items (3 bundles); 2 bidders, 3 items (7 bundles); and 2 bidders, 5 items (31 bundles). 

\paragraph{Baselines} We compare against the VCG mechanism \citep{Vickrey,Clarke,Groves}, previously reported results from RegretNet and RegretFormer \citep{regretformer_paper}, Affine Maximizer Auction trained via grid search (AMA and VVCA), and local search methods (BLAMA, ABAMA, BBBVVCA) \citep{Sandholm_2015}. We re-implement VCG mechanism and the search algorithms \citep{Sandholm_2015}. The search algorithms are trained on 100 samples for 10 different random seeds, and evaluated on 1,000,000 samples.

\paragraph{Performance in non-combinatorial setting} 
We compare the computational results for non-combinatorial setting, where bidders draw their value for each item from $U[0, 1]$ (setting A). In these experiments, we make additive valuation assumption.
The revenue performance is reported in Table \ref{non-combinatorial_results}.

Our revenue performance is comparable to the machine learning-powered models, which outperform the heuristic designs. The slight underperformance in comparison to RegretNet and RegretFormer is due to the complex constrained optimization space. For instance, in the 2-agent, 5-item setting, RegretNet and RegretFormer optimize allocations at the agent-item level (\(2 \times 5\)), while CANet and CAFormer optimize at the agent-bundle level (\(2 \times 31\)), which introduces excessive complexity to non-combinatorial settings, preventing full convergence. Besides, regret estimation is less reliable in larger-scale settings, because adversarial optimization of the inner loss \eqref{inner_loss} might be inaccurate, which affects the convergence of optimal revenue. If the utility of the best misreport is underestimated, the revenue might be overestimated. Future work is encouraged to evaluate this approximation's accuracy.

\begin{table}
\centering
\resizebox{\linewidth}{!}{%
\begin{tblr}{
  width = \linewidth,
  colspec = {Q[250]Q[110]Q[110]Q[110]Q[110]Q[110]Q[110]},
  row{1} = {c},
  row{2} = {c},
  cell{1}{2} = {c=2}{0.22\linewidth},
  cell{1}{4} = {c=2}{0.22\linewidth},
  cell{1}{6} = {c=2}{0.22\linewidth},
  hline{1,3,13} = {-}{},
}
             & 2x2 (A) &       & 2x3 (A) &       & 2x5 (A) &       \\
             & Rev     & Rgt   & Rev     & Rgt   & Rev     & Rgt   \\
VCG          & 0.667   & 0     & 1       & 0     & 1.671   & 0     \\
AMA          & 0.860   & 0     & -       & 0     & -       & 0     \\
VVCA         & 0.866   & 0     & -       & 0     & -       & 0     \\
BLAMA        & 0.786   & 0     & 1.222   & 0     & 2.231   & 0     \\
ABAMA        & 0.786   & 0     & 1.255   & 0     & 2.251   & 0     \\
BBBVVCA      & 0.776   & 0     & 1.238   & 0     & 2.242   & 0     \\
RegretNet    & 0.878   & 0.001 & 1.317   & 0.001 & 2.339   & 0.001 \\
RegretFormer & \textbf{0.908}   & 0.001 & \textbf{1.416}   & 0.001 & \textbf{2.453}   & 0.001 \\
CANet        & 0.879   & 0.001 & 1.317   & 0.001 & 2.282   & 0.004 \\
CAFormer     & 0.891   & 0.001 & 1.326   & 0.001 & 2.329   & 0.005 
\end{tblr}
}
\caption{Non-combinatorial results: Our revenue performance is comparable to that of RegretNet and RegretFormer, which outperform heuristic designs.}
\label{non-combinatorial_results}
\end{table}

\paragraph{Performance in combinatorial setting} We compare the results in combinatorial setting.
Let $v_{ij}^{\text{item}}$ be the valuation of of bidder $i$ for item $j$, individually. The valuation of bidder $i$ for bundle $S$ is drawn as \( v_{iS} = \sum_{j\in S} v_{ij}^{\text{item}} + c_{iS} \) with \( c_{iS} \sim U[-1, 1] \). For symmetric settings (B),  $v_{ij}^{\text{item}} \sim U[1,2], \forall i \in N$. For asymmetric settings (C), $v_{1j}^{\text{item}} \sim U[1,2] \text{ and } v_{2j}^{\text{item}} \sim U[1,5]$.

The results in Table \ref{combinatorial_results} show that our models outperform heuristic benchmarks and RegretNet in all settings. As we specify the regret target as a proportion of revenue, the regret is higher in larger-scale settings. Both CAFormer and CANet achieve negligible regret in all experiments. However, CAFormer, with a greater number of parameters and expressivity, consistently outperforms CANet in revenue.

\section{Discussion and Future Work} \label{sec:discussion}

We view our approach as an initial step towards generalizing differentiable economics for combinatorial auction, though it remains an open theoretical question whether it can fully represent the true optimal mechanism. Although most revenue-maximizing heuristic methods struggle with combinatorial explosion in larger-scale settings, our approach offers greater efficiency. While our formulation and experiments focused on valuations with known distributions (particularly uniform) and considered all possible bundles, by design, this approach is inherently not limited to assumptions about the structures of bundles and valuation profiles. It can even be applied to online auctions and auctions whose distribution is unknown.

In practice, if items are divisible, fractional allocations directly represent feasible outcomes; however, if they are not, implementing the lottery in large-scale settings can be challenging. In some cases, initial use of randomized mechanisms can enable convergence to effective solutions, which can subsequently be fine-tuned to obtain deterministic results.

Of the two proposed architectures, CAFormer not only outperforms in revenue, but also offers other advantages: permutation-equivariance, the ability to generalize to unseen settings with variable input size, and to incorporate contextual information. Moreover, we offer a technical perspective, which considers the multi-task nature of the problem, providing flexibility in training. However, the non-convexity of the min-max optimization poses a challenge for gradient-based methods, as they lack theoretical guarantees for convergence. This limits the power of differentiable economics in larger-scale problems. Exploring alternative loss-balancing techniques could yield further insights, but this is beyond the scope of our current work and is left for future investigation. 

Although our focus is on auctions, we expect that the proposed approach could extend to other differentiable combinatorial optimization problems, as long as the feasible region can be parameterized for gradient-based learning and the objective function remains differentiable. Developing techniques to incorporate constraints more effectively in gradient-based optimization is a promising future research direction.

\clearpage

\bibliographystyle{named}
\bibliography{ijcai25}

\end{document}